\documentclass[american,aps,pra,reprint, superscriptaddress]{revtex4-1}
\usepackage[unicode=true,pdfusetitle, bookmarks=true,bookmarksnumbered=false,bookmarksopen=false, breaklinks=false,pdfborder={0 0 0},backref=false,colorlinks=false] {hyperref}
\hypersetup{ colorlinks,linkcolor=myurlcolor,citecolor=myurlcolor,urlcolor=myurlcolor}
\usepackage{graphics,epstopdf,graphicx, amsthm, amsmath, amssymb, times, braket, colortbl, color, bm, framed, cleveref, mathrsfs}
\usepackage[up]{subfigure}
\definecolor{myurlcolor}{rgb}{0,0,0.7}

\newcommand{\tinyspace}{\mspace{1mu}}

\newcommand{\proj}[1]{| #1\rangle\!\langle #1 |}

\newcommand{\iinner}[2]{\langle #1 | #2\rangle}

\newcommand{\Pa}[1]{\left(#1\right)}
\newcommand{\Br}[1]{\left[#1\right]}
\newcommand{\abs}[1]{\left\lvert\tinyspace #1 \tinyspace\right\rvert}
\newcommand{\norm}[1]{\left\lVert #1 \right\rVert}

\theoremstyle{plain}
\newtheorem{thm}{Theorem}

\newtheorem{prop}[thm]{Proposition}
\newtheorem{cor}[thm]{Corollary}

\newcommand*{\myproofname}{Proof}

\def\cG{\mathcal{G}}
\def\cM{\mathcal{M}}
\def\cN{\mathcal{N}}

\def\sD{\mathscr{D}}

\def\cH{\mathcal{H}}

\def\diff{\mathrm{d}}
\def\md{{(\mathrm{d})}}

\makeatother

\begin{document}

 \author{Kaifeng Bu}
 \email{bkf@zju.edn.cn}
 \affiliation{School of Mathematical Sciences, Zhejiang University, Hangzhou 310027, PR~China}
 \author{Uttam Singh}
 \email{uttamsingh@hri.res.in}
 \affiliation{Harish-Chandra Research Institute, Allahabad, 211019, India}
 \author{Lin Zhang}
 \email{linyz@zju.edu.cn}
 \affiliation{Institute of Mathematics, Hangzhou Dianzi University, Hangzhou 310018, PR China}
  \author{Junde Wu}
 \email{wjd@zju.edn.cn}
 \affiliation{School of Mathematical Sciences, Zhejiang University, Hangzhou 310027, PR~China}

\title{Average distance of random pure states from the maximally entangled and coherent states }

\begin{abstract}
It is well known that random bipartite pure states are typically maximally entangled within an arbitrarily small error. Showing that the marginals of random bipartite pure states are typically extremely close to the maximally mixed state, is a way to prove the above. However, a more direct way to prove the above is to estimate the distance of random bipartite pure states from the set of maximally entangled states. Here, we find the average distance between a random bipartite pure state and the set of maximally entangled states as quantified by three different quantifiers of the distance and investigate the typical properties of the same. We then consider random pure states of a single quantum system and give an account of the typicality of the average $l_1$ norm of coherence for these states scaled by the maximum value of the $l_1$ norm of coherence. We also calculate the variance of the $l_1$ norm of coherence of random pure states to elaborate more on the typical nature of
 the scaled average $l_1$ norm of coherence. Moreover, We compute the distance of a random pure state from the set of maximally coherent states and obtain the average distance.
\end{abstract}

\maketitle

\section{Introduction}
The typicality of various quantities such as the quantum entanglement \cite{Hayden2004, Hayden2006, Collins2015} and the diagonal entropy (the entropy of the diagonal part of a density matrix in a fixed basis) \cite{Giroad2016} has been proved very fruitful. In particular, the existence of typicality helps in the reduction of computational complexity of these quantities which is extremely relevant for the systems having higher dimensional Hilbert spaces. Also, the typicality of quantum entanglement of random bipartite pure states sampled from the uniform Haar measure provides a satisfactory explanation to the {\it equal a priori} postulate of the statistical physics \cite{Popescu2006, Goldstein2006}. Recently, in the context of metrology \cite{Giovannetti2004, Giovannetti2006, Giovannetti2011, Toth2014}, it has been proved that random pure states sampled from the uniform Haar measure typically do not lead to the super-classical scaling of precision, however, random states from the symmetric subspace typically lead to the optimal Heisenberg scaling \cite{Oszmaniec2016}. The key ingredient in the proof of all these results is the {\it concentration of measure phenomenon}, in particular L\'evy's lemma \cite{Ledoux2005}.
The remarkable result that random bipartite pure states are typically maximally entangled within vanishingly small error has been obtained by proving that the reduced density matrices corresponding to random bipartite pure states are typically very close to the maximally mixed state. However, here we present a direct way to investigate this feature of random bipartite pure states by finding the distance of Haar distributed bipartite pure states from the set of maximally entangled states. This results in various finer details of the problem of the typicality of quantum entanglement. We also find a complementarity type relationship between this distance and the negativity--a computable measure of entanglement--of Haar distributed bipartite pure states and further delineate the typicality of negativity.

Recent advances in the fields of quantum thermodynamics \cite{Aspuru13, Horodecki2013, Skrzypczyk2014, Narasimhachar2015, Brandao2015, Rudolph214, Rudolph114, Piotr2015, Gardas2015, Avijit2015, Goold2015} and quantum biology
\cite{Plenio2008,Levi14,Aspuru2009, Lloyd2011, Li2012, Huelga13} suggest the crucial role of quantum superposition in various physical processes and have led to the development of various resource theories of quantum coherence \cite{Gour2008, Marvian14, Baumgratz2014, AlexB2015, Chitambar2016}. This has ignited a great deal of interest towards proposing various measures of quantum coherence\cite{Baumgratz2014, Napoli2016, Piani2016, Killoran2015,Girolami14, Chitambar2015, StreltsovA2015, ChitambarA2015, Peng2015}, establishing its connection to the other characteristic traits of quantum physics such as entanglement \cite{Alex15} and providing operational meaning to the quantifiers of coherence \cite{Winter2015, UttamA2015}. In the spirit of finding the typical coherence content of random quantum states, the average relative entropy of coherence of random pure and mixed quantum states sampled from various probability measures has been calculated and shown to be typical \cite{UttamS2015, ZhangLin2015}. The immense importance of these results has been shown explicitly in estimating analytically  the typical entanglement content of a specific class of random bipartite mixed states, which is an extremely difficult task for higher dimensional quantum systems \cite{ZhangLin2015}. However, not much is known of the typicality of the $l_1$ norm of coherence, which is one of the computable measures of quantum coherence \cite{Baumgratz2014}, of the Haar-distributed pure quantum states.
In this work,
we calculate the expected value and the variance of $l_1$ norm of coherence of Haar distributed pure sates of a single quantum system and show that the average $l_1$ norm of coherence does not show the typicality, however, the average $l_1$ norm of coherence scaled by its maximum value does show the typicality. In particular, we find that the ratio of the average $l_1$ norm of coherence and its maximum value for most of the Haar distributed pure quantum states is concentrated around the fixed value $\pi/4$. Furthermore, similar to the entanglement case, we investigate the average distance between a random pure state and the set of the maximally coherent states with respect to  three different distance measures.

The paper is organized as follows. We start with giving an exposition of the quantum coherence and other preliminaries in Sec. \ref{sec:prelims}. In Sec. \ref{sec:aver_ent} we prove the concentration of negativity and calculate the average distance between random pure states and the set of maximally entangled states employing three different distance measures. We calculate the expected value and variance of the $l_1$ norm of coherence for Haar distributed pure quantum states in Sec. \ref{sec:averl1}. In Sec. \ref{sec:aver_dis} we find the average value of the distance between Haar distributed pure quantum states and the set of maximally coherent states. Furthermore, we calculate the average value of the $\alpha$-classical purity and show its typicality  in  Sec. \ref{aver_puri}. Finally, we conclude in Sec. \ref{sec:con} with a discussion on the results presented in this work.

\section{Preliminaries}
\label{sec:prelims}
Here we present the relevant basic tools and concepts that are required for presenting our main results.

\noindent
{\it Quantum coherence.--} In the pursuit of finding natural restricted class of operations (allowed operations) for an operationally motivated resource theory of coherence, various resource theories of coherence have been developed very recently \cite{Gour2008, Marvian14, Baumgratz2014, AlexB2015, Chitambar2016}. We consider here a measure of coherence, namely, the $l_1$ norm of coherence introduced in the resource theory of coherence based on the incoherent operations \cite{Baumgratz2014}. For a $d$ dimensional quantum system in a state $\rho$ and a fixed reference basis $\{\ket{i}\}$, the $l_1$ norm of coherence $C_{l_1}(\rho)$ is defined as
\begin{align}
C_{l_1}(\rho) := \sum_{\substack{{i,j=1}\\{i\neq j}}}^d |\bra{i}\rho\ket{j}|.
\end{align}
We emphasize that the notion of coherence is intrinsically basis dependent and applicable to finite dimensional systems.


\smallskip
\noindent
{\it L\'evy's lemma (see \cite{Ledoux2005} and \cite{Hayden2006}).--}
Let $\cG:\mathbb{S}^k\to \mathbb{R}$ be a Lipschitz continuous function from the $k$-sphere to the real line with a Lipschitz constant $\eta$ (with respect to the Euclidean norm). Let a point $z\in\mathbb{S}^k$ is chosen uniformly at random. Then for any $\varepsilon>0$,
\begin{eqnarray}
\label{eq:levy-lemma}
\mathrm{Pr}\set{|\cG(z)-\mathbb{E}\cG|>\varepsilon} \leq 2\exp\Pa{-\frac{(k+1)\varepsilon^2}{9\pi^3\eta^2\ln2}},
\end{eqnarray}
where $\mathbb{E}\cG$ is the expected value of $\cG$.

\noindent
\smallskip
{\it Measures of distance (see \cite{Wilde13}).--} We consider three different distance measures, namely, the trace distance, the Hilbert-Schmidt distance and the Bures distance for our purposes. (1) The trace distance between two quantum states $\rho$ and $\sigma$ is defined as
\begin{eqnarray}
\label{tr-dis}
||\rho-\sigma||_1:=\mathrm{Tr}\left[\sqrt{(\rho-\sigma)^2}\right].
\end{eqnarray}
(2) The Hilbert-Schmidt distance between two quantum states $\rho$ and $\sigma$, induced by the
Hilbert-Schmidt scalar product, is defined as
\begin{eqnarray}
\label{hs-dis}
||\rho-\sigma||_2:=\sqrt{\mathrm{Tr}\left[(\rho-\sigma)^2\right]}.
\end{eqnarray}
(3) The Bures distance between two quantum states $\rho$ and $\sigma$ is defined as
\begin{eqnarray}
\label{b-dis}
D_\mathrm{B}(\rho,\sigma)=\sqrt{2\left(1-\mathrm{Tr}\left[(\rho^{1/2}\sigma\rho^{1/2})^{1/2}\right]\right)}.
\end{eqnarray}

\smallskip
\noindent
{\it The Dirichlet integral.--} The Dirichlet distribution \cite{Frigyik2010} of order $N$ with parameters $\vec{\alpha}=(\alpha_1,\alpha_2,...,\alpha_N)$, denoted by $\mathrm{Dir}(\vec{\alpha})$, has a probability density function
\begin{eqnarray}
p(x_1,...,x_N;\vec{\alpha})=\frac{1}{C(\alpha)}\prod^N_{i=1}x^{\alpha_i-1}_i
\end{eqnarray}
on the simplex
\begin{eqnarray*}
\Delta_{N-1}=\set{(x_1,...,x_N)|\sum^N_{i}x_i=1, x_i\geq 0~ \mathrm{for}~1\leq i\leq N }.
\end{eqnarray*}
The normalization constant $C(\alpha)=\frac{\prod^N_{i=1}\Gamma(\alpha_i)}{\Gamma\left(\sum^N_{i=1}\alpha_i\right)}$ comes from the Dirichlet integral \cite{Andrews1999} on the simplex $\Delta_{N-1}$, which is given by
\begin{eqnarray}
\int_{\Delta_{N-1}}x^{\alpha_1-1}
\cdots x^{\alpha_N-1}_N
dx_1\cdots dx_N
=\frac{\prod^N_{i=1}\Gamma(\alpha_i)}{\Gamma\left(\sum^N_{i=1}\alpha_i\right)}.
\end{eqnarray}

\section{Distance of a random pure state from the set of maximally entangled states }
\label{sec:aver_ent}
Consider a bipartite quantum system $\cH_A\otimes \cH_B$ ($\mathrm{dim} \cH_A=\mathrm{dim} \cH_B=N$) in a state $\ket{\psi}_{AB}$. The state admits the Schmidt decomposition as $\ket{\psi}_{AB}=\sum^N_{i=1}\sqrt{\lambda_i}\ket{\phi_i}_A\otimes\ket{\varphi_i}_B$. If all $\lambda_i=\frac{1}{N}$, then such a state is called maximally entangled state. Denote by $\cM_E$, the set of all maximally entangled pure states. Now, we consider the distance between a pure state $\psi_{AB}=\proj{\psi_{AB}}$ and the set $\cM_E$ of the maximally entanglement states. Using the distance measures, given in Eqs. \eqref{tr-dis}, \eqref{hs-dis} and \eqref{b-dis}, we have
\begin{align}
\label{eq:tr-dis1}
\sD_{\mathrm{Tr}}(\psi_{AB},\cM_E)&:=\inf_{\phi_{AB}\in\cM_E}||\psi_{AB}-\phi_{AB}||_1\nonumber\\
&=2\sqrt{1-\frac{1}{N}\left(\sum^N_{i=1}\sqrt{\lambda_i}\right)^2}.
\end{align}
\begin{align}
\label{eq:hs-dis1}
\sD_{\mathrm{HS}}(\psi_{AB},\cM_E)&:=\inf_{\phi_{AB}\in\cM_E}||\psi_{AB}-\phi_{AB}||_2\nonumber\\
&=\sqrt{2}\sqrt{1-\frac{1}{N}\left(\sum^N_{i=1}\sqrt{\lambda_i}\right)^2}\nonumber\\
&=\frac{1}{\sqrt{2}}\sD_{Tr}(\psi_{AB},\cM_E).
\end{align}
\begin{align}
\label{eq:b-dis1}
\sD_\mathrm{B}(\psi_{AB},\cM_E)&:=\inf_{\phi_{AB}\in\cM_E}D_B(\psi_{AB},\phi_{AB})\nonumber\\
&=\sqrt{2}\sqrt{1-\frac{1}{\sqrt{N}}\left(\sum^N_{i=1}\sqrt{\lambda_i}\right)}.
\end{align}

The above distances show complementarity with a computable measure of entanglement, namely, the negativity \cite{Vidal2002}. For a bipartite state $\rho_{AB}$, the negativity across the bipartition $A:B$ is defined as
\begin{eqnarray}
\cN(\rho_{AB})=\frac{\norm{\rho^{T_B}_{AB}}_1-1}{2},
\end{eqnarray}
where $T_{B}$ denotes the  partial transpose with respect to the subsystem $B$.
For a pure state $\ket{\psi}_{AB}=\sum^N_{i=1}\sqrt{\lambda_i}\ket{\phi_i}_A\otimes\ket{\varphi_i}_B$, the negativity is given by
\begin{eqnarray}
\cN(\psi_{AB})=\frac{\Pa{\sum^N_{i=1}\sqrt{\lambda_i}}^2-1}{2}.
\end{eqnarray}
The maximum value of negativity, denoted by $\cN_{max}$, is equal to $\frac{N-1}{2}$. Now, we have (see also \cite{Puchala2015})
\begin{eqnarray}
\label{dual:3}N\frac{\sD^2_{\mathrm{Tr}}(\psi_{AB},\cM_E)}{8}+\cN(\psi_{AB})=\cN_{max};\\
\label{dual:4}N\frac{\sD^2_{\mathrm{HS}}(\psi_{AB},\cM_E)}{4}+\cN(\psi_{AB})=\cN_{max}.
\end{eqnarray}
The meaning of the complementary relations (Eqs. \eqref{dual:3} and \eqref{dual:4}) is easy to grasp: the larger the negativity of a state $\ket{\psi}_{AB}$, the closer it is to the maximally entangled state and vice versa. Next we discuss the concentration of measure phenomenon for the negativity of random bipartite pure quantum states.
\smallskip
\subsection{Concentration of the negativity around its average value}
The average value of negativity $\mathbb{E}_{\psi_{AB}}\cN=\int  \cN(\psi_{AB})d \psi_{AB}$ over Haar distributed pure quantum states has been calculated in Ref. \cite{Datta2010} and it is demonstrated numerically that in the limit $N\to\infty$ it is a constant multiple of the maximal negativity, i.e.,
 \begin{eqnarray}\label{eq:av_neg}
 \mathbb{E}_{\psi_{AB}}\cN\sim 0.72037 \cN_{max}.
 \end{eqnarray}
Moreover, in Ref. \cite{Datta2010}, the numerical evidence of concentration of (scaled) negativity around its average value is provided. Here we give an analytical proof of the concentration of (scaled) negativity based on the concentration of measure phenomenon \cite{Ledoux2005}. But before we proceed further, we find the Lipschitz constant for the negativity in the following.
\begin{prop}\label{prop:N1}
Let $\ket{\psi}_{AB}\in \cH_A\otimes \cH_B$ with $\mathrm{dim} \cH_{A}=N=\mathrm{dim} \cH_{B}$. Then, the Lipschitz constant for the function $\cN:\ket{\psi}_{AB}\to \cN(\psi_{AB})$ is less than $N/2\sqrt{2}$, i.e., for two states $\ket{\psi}_{AB}\in \cH_A\otimes \cH_B$ and $\ket{\phi}_{AB}\in \cH_A\otimes \cH_B$
\begin{eqnarray}\label{eq:lps2}
|\cN(\psi_{AB})-\cN(\phi_{AB})|\leq \frac{N}{2\sqrt{2}}||\psi_{AB}-\phi_{AB}||_2.
\end{eqnarray}
\end{prop}
\begin{proof}
Suppose the Schmidt coefficients of $\psi_{AB}=\proj{\psi_{AB}}$ and $\phi_{AB}=\proj{\phi_{AB}}$ are $\set{\lambda_i}^N_{i=1}$ and $\set{\mu_i}^N_{i=1}$,
respectively.
Based on the definition of negativity, we have
\begin{align*}
&\left|\cN(\psi_{AB})-\cN(\phi_{AB})\right|\\
=&\frac{1}{2}\left|[\mathrm{Tr}\sqrt{\rho_A}]^2-[\mathrm{Tr}\sqrt{\sigma_A}]^2\right|\\
=&\frac{1}{2}\left|[\mathrm{Tr}\sqrt{\rho_A}]^2-[\mathrm{Tr}\sqrt{U\sigma_AU^\dag}]^2\right|,
\end{align*}
where $\rho_A$ and $\sigma_A$ are the reduced states of $\psi_{AB}$ and $\phi_{AB}$, respectively, and $U$ is a unitary operator. Note that $F(\rho, \mathbb{I}/N)=\frac{1}{\sqrt{N}}\mathrm{Tr}(\sqrt{\rho})$, where $F(\rho, \sigma):=\mathrm{Tr}(\sqrt{\sqrt{\rho}\sigma \sqrt{\rho}})$ is the fidelity between states $\rho$ and $\sigma$. This implies that
\begin{align}
&|\cN(\psi_{AB})-\cN(\phi_{AB})|\nonumber\\
=&\frac{N}{2}\left|F(\rho_A, \mathbb{I}/N)^2-F(U\sigma_AU^\dag, \mathbb{I}/N)^2\right|\nonumber\\
\label{eq:neg1}\leq& \frac{N}{2}\sqrt{1-F(\rho_A, U\sigma_AU^\dag)^2},
\end{align}
where the last inequality follows from the fact that \cite{Rastegin2003,Zhang2015qip}
\begin{align*}
\left|F(\rho, \tau)^2-F(\sigma, \tau)^2\right|
\leq \sqrt{1-F(\rho, \sigma)^2},
\end{align*}
for any states $\rho$, $\sigma$ and $\tau$.
From Lemma 1 of Ref. \cite{Vidal2000b} and Uhlmann's theorem \cite{Uhlmann1976,Jozsa1994}, we have
\begin{eqnarray*}
\max_{U}F(\rho_A, U\sigma_AU^\dag)=\sum^N_{i=1}\sqrt{\lambda^{\downarrow}_i\mu^{\downarrow}_i},\\
|\iinner{\psi_{AB}}{\phi_{AB}}|\leq \sum^N_{i=1}\sqrt{\lambda^{\downarrow}_i\mu^{\downarrow}_i},
\end{eqnarray*}
where $\lambda^{\downarrow}_i$ and $\mu^{\downarrow}_i$ denote that
$\set{\lambda_i}^N_{i=1}$ and $\set{\mu_i}^N_{i=1}$ are listed in the decreasing order.
Thus, there exists a unitary operator $U$ such that
$F(\rho_A, U\sigma_AU^\dag)\geq |\iinner{\psi_{AB}}{\phi_{AB}}|$.
Therefore,
using  Eq. \eqref{eq:neg1} and the fact that $||\psi_{AB}-\phi_{AB}||_2=\sqrt{2-2|\iinner{\psi_{AB}}{\phi_{AB}}|^2}$, we have
\begin{eqnarray}
|\cN(\psi_{AB})-\cN(\phi_{AB})|\leq \frac{N}{2\sqrt{2}}||\psi_{AB}-\phi_{AB}||_2.
\end{eqnarray}
This completes the proof of the proposition.
\end{proof}
Now, applying L\'evy's lemma (see Eq. \eqref{eq:levy-lemma}) to the function $\cN^{s}:\ket{\psi}_{AB}\to \cN(\psi_{AB})/\cN_{max}:=\cN^{s}(\psi_{AB})$ and using the above proposition, we have
\begin{align}
&\mathrm{Pr} \Set{\left|\cN^{s}(\psi_{AB})-\mathbb{E}_{\psi_{AB}}\cN^{s} \right|>\varepsilon}\nonumber\\
&~~~~~~~~~~~~~~~~~~~~~~~~~~~~~~~~~\leq  2\exp\Pa{-\frac{16(N-1)^2\varepsilon^2}{9\pi^3\ln2}}.
\end{align}
The above equation is a statement of concentration of (scaled) negativity around its average value. Together with Eq. \eqref{eq:av_neg} we have that for large values of $N$, most of the random bipartite pure states have (scaled) negativity equal to $0.72037$ within an arbitrarily small error. Next, we compute the average values of distances given by Eqs. \eqref{eq:tr-dis1}, \eqref{eq:hs-dis1} and \eqref{eq:b-dis1}.

\subsection{Typicality of the average values of \texorpdfstring{$\sD^2_{\mathrm{Tr}}(\psi_{AB},\cM_E)$, $\sD^2_{\mathrm{HS}}(\psi_{AB},\cM_E)$ and $\sD^2_\mathrm{B}(\psi_{AB},\cM_E)$}{}}
The average values of the distances $\mathbb{E}_{\psi_{AB}}\sD^2_{\mathrm{Tr}}(\psi_{AB},\cM_E)$ and $\mathbb{E}_{\psi_{AB}}\sD^2_{\mathrm{HS}}(\psi_{AB},\cM_E)$ can be obtained by integrating the complementarity relations \eqref{dual:3} and \eqref{dual:4}, respectively. Now using Eq. \eqref{eq:av_neg}, we have
\begin{eqnarray*}
\mathbb{E}_{\psi_{AB}}\sD^2_{\mathrm{Tr}}(\psi_{AB},\cM_E)\sim  1.1185,\\
\mathbb{E}_{\psi_{AB}}\sD^2_{\mathrm{HS}}(\psi_{AB},\cM_E)\sim   0.5593.
\end{eqnarray*}
To establish typicality of $\mathbb{E}_{\psi_{AB}}\sD^2_{\mathrm{Tr}}(\psi_{AB},\cM_E)$ and $\mathbb{E}_{\psi_{AB}}\sD^2_{\mathrm{HS}}(\psi_{AB},\cM_E)$, we need Lipschitz constants for the functions $\cG_1:\ket{\psi}_{AB}\mapsto \sD^2_{\mathrm{Tr}}(\psi_{AB},\cM_E):=\cG_1(\psi_{AB})$ and $\cG_2:\ket{\psi}_{AB}\mapsto \sD^2_{\mathrm{HS}}(\psi_{AB},\cM_E):=\cG_2(\psi_{AB})$, respectively. Let us consider two bipartite pure states $\ket{\psi}_{AB}$ and $\ket{\phi}_{AB}$. Note that
\begin{align}
|\cG_1(\psi_{AB})-\cG_1(\phi_{AB})| &= \frac{8}{N}|\cN(\psi_{AB})-\cN(\phi_{AB})|\nonumber\\
&\leq 2\sqrt{2}||\psi_{AB}-\phi_{AB}||_2.
\end{align}
Here in first line, we have used the complementarity relation given by Eq. \eqref{dual:3} and in the second line we have used Proposition \ref{prop:N1}. Similarly,
\begin{align}
|\cG_2(\psi_{AB})-\cG_2(\phi_{AB})| &= \frac{4}{N}|\cN(\psi_{AB})-\cN(\phi_{AB})|\nonumber\\
&\leq \sqrt{2}||\psi_{AB}-\phi_{AB}||_2.
\end{align}
Having found the Lipschitz constants for the functions $\cG_1(\psi_{AB})$ and $\cG_2(\psi_{AB})$, we apply L\'evy's lemma to the functions $\cG_1(\psi_{AB})$ and $\cG_2(\psi_{AB})$. This establishes that $\sD^2_{\mathrm{Tr}}(\psi_{AB},\cM_E)$ and $\sD^2_{\mathrm{HS}}(\psi_{AB},\cM_E)$ concentrate around their respective average values as $N\to\infty$. As it is a proven fact that Haar distributed bipartite states are maximally entangled states \cite{Hayden2004}, it is a bit surprising that in $N\to\infty$ most of the states maintain a finite nonzero distance from the set of maximally entangled states in contrast to our expectations that most states should have this distance very close to zero. However, the point that in $N\to\infty$ most of the states maintain a finite nonzero distance from the set of maximally entangled state makes sense in the view that for most of the random bipartite pure states (scaled) negativity doesn't tend to one. This may mean that the distances cons
 idered i
 n this work explains the entanglement features as seen by the negativity.

Now, we compute the expected value of $\sD^2_\mathrm{B}(\psi_{AB},\cM_E)$. The joint distribution of Schmidt coefficients $\Lambda=(\lambda_1,...,\lambda_N)$ for a bipartite system with subsystems being of equal dimensions is given by \cite{Lloyd1988}
\begin{eqnarray*}
P(\mathbf{\Lambda})d\mathbf{\Lambda}=C \delta(1-\sum^N_{i=1}\lambda_i)
\prod_{1\leq i<j\leq N}(\lambda_i-\lambda_j)^2\prod^N_{k=1}d\lambda_k,
\end{eqnarray*}
where $C$ is the normalization constant, $\delta$ is the Dirac delta
function and $d\mathbf{\Lambda}=\prod^N_{k=1}d\lambda_k$.
Since $\sD^2_\mathrm{B}(\psi_{AB},\cM_E)=2\Br{1-\frac{1}{\sqrt{N}}(\sum^N_{i=1}\sqrt{\lambda_i})}$, we only need to calculate
\begin{eqnarray}
\int \sum^N_{i=1}\sqrt{\lambda_i}P(\mathbf{\Lambda})d\mathbf{\Lambda}.
\end{eqnarray}
We introduce new variables $q_i=t\lambda_i$ such that \cite{Datta2010,Sen1996,Scott2003}
\begin{eqnarray}\label{eq:Q}
Q(\mathbf{q})d\mathbf{q}&\equiv&\prod_{1\leq i< j\leq N} (q_i-q_j)^2
\prod^N_{m=1}e^{-q_m}dq_{m}\nonumber\\
&=&\overline{C} e^{-t}t^{N^2-1}P(\mathbf{\Lambda})d\mathbf{\Lambda}dt,
\end{eqnarray}
where $q_i\in[0,\infty)$ and $t=\sum^N_{i=1}q_i$. Integrating both sides of Eq. \eqref{eq:Q}, we obtain that the normalization constant $\overline{C}$ is equal to $\overline{Q}/\Gamma(N^2)$,  where $\overline{Q}\equiv \int Q(\mathbf{q})d\mathbf{q}=N!\prod^N_{m=1}
\Gamma(m)^2$  \cite{Datta2010,Scott2003}. Thus,
\begin{align}\label{eq:QeqivP}
\int\sum^N_{i=1}\sqrt{q_i}Q(\mathbf{q})d\mathbf{q}
=\overline{Q}\frac{\Gamma(N^2+1/2)}{\Gamma(N^2)}
\int\sum^N_{i=1}\sqrt{\lambda_i}P(\mathbf{\Lambda})d\mathbf{\Lambda}.
\end{align}
The product $\prod_{1\leq i< j\leq N} (q_i-q_j)^2$ is the square of the Van der Monde determinant \cite{Datta2010,Sen1996,Scott2003}
\begin{gather}\label{eq:det}
\prod_{1\leq i<j\leq N}(q_i-q_j)
=\left|
\begin{array}{cccc}
1&\cdots &1\\
q_1 & \cdots & q_N\\
\vdots & \ddots & \vdots\\
q^{N-1}_1 & \cdots & q^{N-1}_N
\end{array}
\right|\nonumber
\\
\label{eq:det}=c_N\left|
\begin{array}{cccc}
\Gamma(1)L_0(q_1)&\cdots &\Gamma(1)L_0(q_N)\\
\Gamma(2)L_1(q_1) & \cdots &\Gamma(2)L_1(q_N)\\
\vdots & \ddots & \vdots\\
\Gamma(N)L_{N-1}(q_1) & \cdots & \Gamma(N)L_{N-1}(q_N)
\end{array}
\right|.
\end{gather}
The second determinant in Eq. \eqref{eq:det} is due to the fact that the determinant does not change when we add a multiple of one row to another. The factor $c_N=\pm 1$  depends on $N$. The polynomials $L_k(x)$ are Laguerre polynomials \cite{Andrews1999}, defined as
\begin{eqnarray}
L_k(x):=\frac{e^x}{k!}\frac{d^k}{dx^k}(e^{-x}x^k), ~for~k\geq 0.
\end{eqnarray}
The Laguerre polynomials satisfy
\begin{eqnarray}
\int^\infty_0 dx e^{-x}L_i(x)L_j(x)=\delta_{ij}.
\end{eqnarray}

Now, we are ready to compute the following integral:
\begin{eqnarray*}
&&\int \sqrt{q_i}Q(\mathbf{q})d\mathbf{q}\\
&=&\int \sqrt{q_i}\prod_{1\leq i< j\leq N} (q_i-q_j)^2\prod^N_{m=1}dq_m\\
&=&\sum_{\sigma,\tau\in S_N}(-1)^{sgn(\sigma)+sgn(\tau)}\\
&\times& \Gamma(\sigma(i))\Gamma(\tau(i))\int dq_i \sqrt{q_i}e^{-q_i}L_{\sigma(i)-1}(q_i)L_{\tau(i)-1}(q_i)\\
&\times&\prod^{N}_{m\neq i}\Gamma(\sigma(m))\Gamma(\tau(m))\int dq_m e^{-q_m}L_{\sigma(m)-1}(q_m)L_{\tau(m)-1}(q_m)\\
&=&\sum_{\sigma\in S_N}\prod^N_{m=1}\Gamma(\sigma(m))^2\int d q_i\sqrt{q_i}e^{-q_i}L_{\sigma(i)-1}(q_i)L_{\sigma(i)-1}(q_i)\\
&=&(N-1)!\prod^N_{m=1}\Gamma(m)^2\Pa{\sum^N_{k=1}I^{(1/2)}_{kk}},
\end{eqnarray*}
where $S_N$ is the permutation group on $\set{1,...,N}$ and $I^{(\alpha)}_{ij}$ is defined as \cite{Datta2010,Scott2003}
\begin{eqnarray}
I^{(\alpha)}_{ij}=\int^\infty_0 e^{-x}x^\alpha L_i(x)L_j(x)dx.
\end{eqnarray}
Thus,
\begin{eqnarray*}
\int\sum^N_{i=1}\sqrt{q_i}Q(\mathbf{q})d\mathbf{q}
&=&N!\prod^N_{m=1}\Gamma(m)^2(\sum^N_{k=1}I^{(1/2)}_{kk})\\
&=&\overline{Q}\Pa{\sum^N_{k=1}I^{(1/2)}_{kk}}.
\end{eqnarray*}
It implies that
\begin{align}
\int\sum^N_{i=1}\sqrt{\lambda_i}P(\mathbf{\Lambda})d\mathbf{\Lambda}
=\frac{\Gamma(N^2)}{\Gamma(N^2+1/2)}\Pa{\sum^N_{k=1}I^{(1/2)}_{kk}},
\end{align}
where $I^{(1/2)}_{kk}$ is given by \cite{Datta2010}
\begin{eqnarray}
I^{(1/2)}_{kk}=\frac{(-1)^k}{k!}\sum^k_{n=0}
\Pa{\begin{array}{cc}
k\\
n
\end{array}}
\frac{[\Gamma(n+3/2)]^2}{n!\Gamma(n-k+3/2)}.
\end{eqnarray}
Thus, the average value of $\sD^2_B(\psi_{AB},\cM_E)$ over Haar-distributed pure states is given by
\begin{eqnarray}
\label{eq:bures-dis-2}
2\Br{1-\frac{1}{\sqrt{N}}\frac{\Gamma(N^2)}{\Gamma(N^2+1/2)}\left(\sum^N_{k=1}I^{(1/2)}_{kk}\right)}.
\end{eqnarray}
Consider the map $\cG:\ket{\psi}_{AB}\to \sD^2_B(\psi_{AB}, \cM_E):=\cG(\psi_{AB})$. Then, similar to the reasoning as in the proof of Proposition \ref{prop:N1}, we have
 \begin{align}
 &\left|\sD^2_B(\psi_{AB}, \cM_E)-\sD^2_B(\phi_{AB}, \cM_E)\right|\nonumber\\
 &=2|F(\rho_A, \mathbb{I}/N)-F(U\sigma_AU^\dag, \mathbb{I}/N)|\nonumber\\
\label{eq:neg2}&\leq 2\sqrt{1-F(\rho_A, U\sigma_AU^\dag)^2}\nonumber\\
&\leq \sqrt{2}||\psi_{AB}-\phi_{AB}||_2.
 \end{align}
Thus, the Lipschitz constant for the function $\cG:\ket{\psi}_{AB}\to \sD^2_\mathrm{B}(\psi_{AB}, \cM_E)$ can be taken to be $\sqrt{2}$. Now, using L\'evy's lemma we establish that $\sD^2_\mathrm{B}(\psi_{AB},\cM_E)$ concentrates around its expected value as $N\to\infty$. Again the average distance $\mathbb{E}_{\psi_{AB}}\sD^2_{\mathrm{B}}(\psi_{AB},\cM_E)$ is finite and nonzero in contrast to our expectations (see Figure \ref{fig}).

Furthermore, we ask how close the reduced state $\rho_A$ of $\ket{\psi}_{AB}$ to the maximally mixed state. Let $\mathcal{H}$ be a complex Hilbert space, then for any two positive semidefinite operators $\rho$ and $\sigma$ in $\mathcal{H}$ \cite{Watrous2011}, we have
\begin{eqnarray}
\label{lem:1v2}
\norm{\rho-\sigma}_1\geq \norm{\sqrt{\rho}-\sqrt{\sigma}}^2_2.
\end{eqnarray}
Now, taking $\rho=\rho_A$ where $\rho_A$ is the reduced state of $\ket{\psi}_{AB}$ and $\sigma=\frac{\mathbb{I}}{N}$, we obtain
\begin{eqnarray*}
\norm{\rho_A-\frac{\mathbb{I}}{N}}_1 \geq \sD^2_\mathrm{B}(\psi_{AB},\cM_E).
\end{eqnarray*}
Thus,
\begin{eqnarray*}
\mathbb{E}_{\psi_{AB}}\norm{\rho_A-\frac{\mathbb{I}}{N}}_1\geq \mathbb{E}_{\psi_{AB}}\sD^2_\mathrm{B}(\psi_{AB},\cM_E).
\end{eqnarray*}
The average value of $\sD^2_\mathrm{B}(\psi_{AB},\cM_E)$ over Haar-distributed pure states approaches to a fixed value which is close to $2$. Thus, we see that the average value of $\norm{\rho_A-\frac{\mathbb{I}}{N}}_1$ will not tend to zero in $N\to\infty$ limit.

\begin{figure}
\includegraphics[width=70mm]{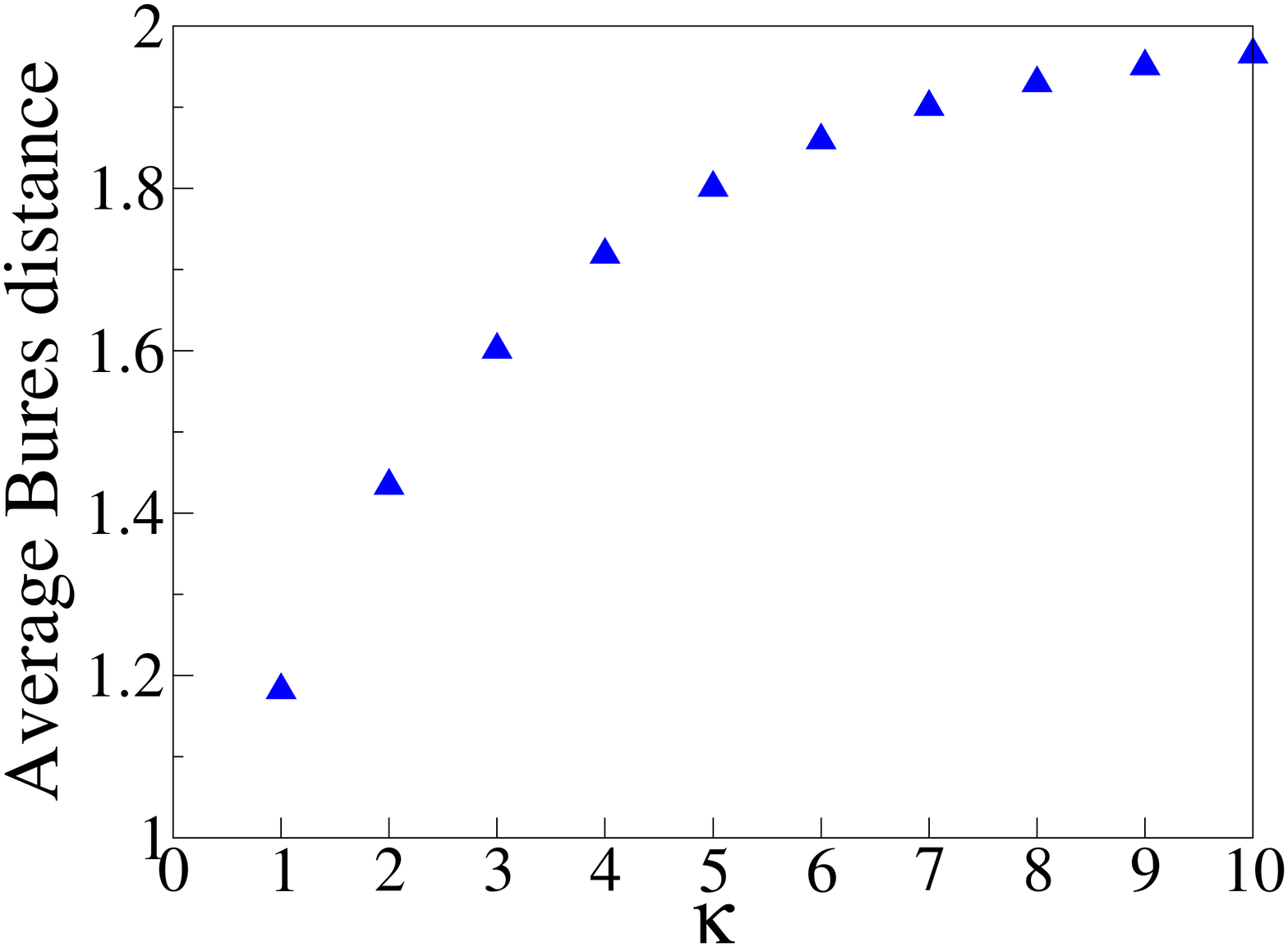}
\caption{The plot of the average Bures distance of random bipartite pure states from the set of maximally entangled states as a function of $\kappa$ such that $N=2^\kappa$. It shows that as we increase $N$, the Bures distance approaches to a fixed value which is close to $2$ (see Eq. \ref{eq:bures-dis-2}).}
\label{fig}
\end{figure}
\section{Typicality of the average $l_1$ norm of coherence of random pure states}\label{sec:averl1}

For a random pure state $\ket{\psi}=\left(\sqrt{\lambda_1}e^{i\theta_1},\ldots, \sqrt{\lambda_N}e^{i\theta_N}\right)^T$, the joint distribution of $\set{\lambda_i}$ will be the uniform distribution on a simplex $\Delta_{N-1}:=\left\{\left( \lambda_1,\ldots,\lambda_N |\sum_{i=1}^N\lambda_i=1, \lambda_i\geq 0~ \mathrm{for}~1\leq i\leq N  \right)\right\}$ and is given by \cite{UttamS2015}
\begin{eqnarray}
\label{eq:prob-dis}
p_c(\mathbf{\Lambda})d\mathbf{\Lambda}=\Gamma(N)\delta\left(1-\sum^N_{i=1}\lambda_i\right)\prod^N_{i=1}\mathrm{d}\lambda_i.
\end{eqnarray}
where $d\mathbf{\Lambda}=\prod^N_{k=1}d\lambda_k$. This distribution is a special Dirichlet distribution \cite{Frigyik2010}. Moreover, moments of such uniform Dirichlet-distributed random variables can be written as \cite{Frigyik2010}
\begin{eqnarray}
\mathbb{E}\Br{\prod^N_{i=1}{\lambda^{\alpha_i}_i}}
=\frac{\Gamma(N)}{\Gamma(N+\sum^N_{i=1}\alpha_i)}
\times \prod^N_{i=1}\Gamma(1+\alpha_i).
\end{eqnarray}
This will be crucial for our calculations. The $l_1$ norm of coherence for the pure state $\ket{\psi}=(\sqrt{\lambda_1}e^{i\theta_1},\ldots, \sqrt{\lambda_N}e^{i\theta_N})^T$ is given by
\begin{eqnarray}
C_{l_1}(\psi)=\sum_{i\neq j}\sqrt{\lambda_i \lambda_j}
=\left(\sum^N_{i=1}\sqrt{\lambda_i}\right)^2-1.
\end{eqnarray}
Here and in the rest parts of the paper we have $\psi=\ket{\psi}\bra{\psi}$.
The average value of the $l_1$ norm of coherence over all pure states, denoted by $\mathbb{E}_\psi C_{l_1}(\psi)$, is defined as $\mathbb{E}_\psi C_{l_1}(\psi):=\int d(\psi)C_{l_1}(\psi)$, which can be easily calculated as
\begin{eqnarray*}
\mathbb{E}_\psi C_{l_1}(\psi)&=&\int d(\psi)C_{l_1}(\psi)\\
&=&\Gamma(N)\sum_{i\neq j} \int \sqrt{\lambda_i}\sqrt{\lambda_j}\delta(1-\sum^N_{i=1}\lambda_i)\prod^N_{i=1}\diff\lambda_i\\
&=&\frac{(N-1)\pi}{4},
\end{eqnarray*}
where we have made use of the Dirichlet's integral (cf. \cite{Puchala2015}).
In the $N$ dimensional case, the maximum value of the $l_1$ norm of coherence is given by $N-1$. Therefore,
\begin{eqnarray}
\frac{\mathbb{E}_\psi C_{l_1}(\psi)}{C^{max}_{l_1}}=\frac{\pi}{4}\approx 0.785398,
\end{eqnarray}
which holds for any dimension $N$. However, the expected value of the relative entropy of coherence over all pure states has been shown to be \cite{UttamS2015}
\begin{eqnarray}
\mathbb{E}_\psi C_r(\psi)=\sum^N_{k=2}1/k.
\end{eqnarray}
When $N$ is very large, $\sum^N_{k=1}1/k\approx \ln(N)+\gamma$, where $\gamma$ is  Euler's constant \cite{Andrews1999}. Besides the maximal relative entropy of coherence of an $N$ dimensional quantum system is $C^{max}_{r}=\ln(N)$. Thus,
\begin{eqnarray}
\frac{\mathbb{E}_\psi C_r(\psi)}{C^{max}_r}\to 1 ~\mathrm{as}~N\to\infty,
\end{eqnarray}
which is very different from the case of $l_1$ norm of coherence. In Ref.\cite{UttamS2015}, the authors have shown the  concentration of measure phenomenon for relative entropy coherence based on the Levy's Lemma. In the following, we show that $l_1$ norm of coherence does not show the concentration phenomenon.

\begin{thm}[No concentration for the $l_1$ norm of coherence]
\label{thm:lip-con}
Let $\ket{\psi}$ be a random pure state in $N$ dimensional Hilbert space, then the probability
that $C_{l_1}(\psi)$ is not close to $\frac{\pi}{4}(N-1)$ is nonzero even for very large $N$, i.e., for any $\varepsilon>0$
\begin{align}
\mathrm{Pr} \Set{\left|C_{l_1}(\psi) -\frac{\pi}{4} (N-1)\right|>\varepsilon}
\leq  2 \exp\Pa{-\frac{4\varepsilon^2}{9\pi^3 N \ln2}}.
\end{align}
\end{thm}

\begin{proof}
We use L\'evy's lemma (Eq. (\ref{eq:levy-lemma})) to prove the theorem. For this, we need to find the Lipschitz constant for the $l_1$ norm of coherence, i.e., $C_{l_1}:\ket{\psi}\rightarrow C_{l_1}(\psi)$,  where $\ket{\psi}\in\mathcal{H}^N$. Let $\ket{\psi}=\sum_{i=1}^{N}\psi_i\ket{i}$ and $\ket{\phi}=\sum_{i=1}^{N}\phi_i\ket{i}$. Now, similar to
Proposition \ref{prop:N1}, it can be shown that
\begin{align}
\label{eq:lip-l1}
\left| C_{l_1}(\psi) - C_{l_1}(\phi) \right| \leq \frac{N}{\sqrt{2}} ||\psi-\phi||_2.
\end{align}

Thus, the Lipschitz constant $\eta$ for the $l_1$ norm of coherence, i.e., $C_{l_1}:\ket{\psi}\rightarrow C_{l_1}(\psi)$,  where $\ket{\psi}\in\mathcal{H}^N$,  is less than $\frac{N}{\sqrt{2}}$. Take $\ket{\psi}=\ket{1}$ and $\ket{\phi}=\frac{1}{\sqrt{N}}\sum^N_{i=1}\ket{i}$, then
$|C_{l_1}(\psi)-C_{l_1}(\phi)|=N-1$ and $||\psi-\phi||_2=\sqrt{2-\frac{2}{N}}$.
Then
$\eta \geq \frac{N-1}{\sqrt{2-\frac{2}{N}}}\sim \frac{1}{\sqrt{2}}N$, when $N$ is large. Thus,
$\eta$ can be taken to be $\frac{N}{\sqrt{2}}$, as $N$ is very large.
Having obtained the Lipschitz constant and the average value of the $l_1$ norm of coherence, we are ready to use L\'evy's lemma (Eq. (\ref{eq:levy-lemma})), which for any $\varepsilon >0$ reads as
\begin{align*}
\mathrm{Pr} \Set{\left|C_{l_1}(\psi)-\frac{\pi}{4}(N-1) \right|>\varepsilon}
\leq  2 \exp\Pa{-\frac{4\varepsilon^2}{9\pi^3 N \ln2}}.
\end{align*}
This completes the proof of the theorem.
\end{proof}
Theorem \ref{thm:lip-con} tells that the probability that the $l_1$ norm of coherence does not concentrate around $(N-1)\pi/4$ is non-vanishing in the limit $N\rightarrow\infty$. Thus, based on Levy's Lemma, we proved that the $l_1$ norm of coherence does not show the concentration phenomenon. However, in the following we show that the scaled $l_1$ norm of coherence $\frac{C_{l_1}(\psi)}{C^{max}_{l_1}}$ concentrates around $\pi/4$ for very large values of $N$. This is consequence of the following corollary to the above theorem.

\begin{cor}\label{cor:con_of_l1}
Let $\ket{\psi}$ be a random pure state in an $N$ dimensional Hilbert space, then the probability
that $\frac{C_{l_1}(\psi)}{C^{max}_{l_1}}$ is not close to $\frac{\pi}{4}$ is bounded from above by an exponentially small number when $N$ is large, i.e., for any $\varepsilon>0$
\begin{align}
\mathrm{Pr} \Set{\left|\frac{C_{l_1}(\psi)}{C^{max}_{l_1}} -\frac{\pi}{4} \right|>\varepsilon}
\leq  2 \exp\Pa{-\frac{4(N-1)^2\varepsilon^2}{9N\pi^3\ln2}}.
\end{align}
\end{cor}
\begin{proof}
From Theorem \ref{thm:lip-con}, the Lipschitz constant for the function $\frac{C_{l_1}(\psi)}{C^{max}_{l_1}}$ is given by $\frac{N}{\sqrt{2}(N-1)}$ as for any two pure states $\ket{\psi}$ and $\ket{\phi}$, we have
\begin{align}
\left|\frac{C_{l_1}(\psi)}{C^{max}_{l_1}} - \frac{C_{l_1}(\phi)}{C^{max}_{l_1}} \right| &\leq \frac{N}{\sqrt{2}(N-1)}||\psi-\phi||_2.
\end{align}
Now, applying L\'evy's lemma to the function $\frac{C_{l_1}(\psi)}{C^{max}_{l_1}}$, the proof follows.
\end{proof}
Corollary \ref{cor:con_of_l1} shows that the probability that $\frac{C_{l_1}(\psi)}{C^{max}_{l_1}}$ is not close to $\frac{\pi}{4}$ is close to zero for sufficiently large $N$. Therefore,  $\frac{C_{l_1}(\psi)}{C^{max}_{l_1}}$ shows concentration phenomenon, although $C_{l_1}(\psi)$ does not show the same. This point can also be seen from the calculation of the variance $\mathbb{D}_\psi C_{l_1}(\psi)$ of the $l_1$ norm of coherence of Haar distributed pure states, where
\begin{eqnarray}
\mathbb{D}_\psi C_{l_1}(\psi) = \mathbb{E}_\psi C^2_{l_1}(\psi)
- (\mathbb{E}_\psi C_{l_1}(\psi))^2.
\end{eqnarray}
Based on the following formula \cite{Datta2010}, $
\Pa{\sum^N_{i=1}\sqrt{\lambda_i}}^4=1 + 2\sum^N_{\substack{i,j=1\\i\neq j}}\left(\sqrt{\lambda_i\lambda_j} + \lambda_i\lambda_j\right)
+4\sum^N_{\substack{i,j,k=1\\i\neq j\neq k}}\lambda_i\sqrt{\lambda_j\lambda_k}+
\sum^N_{\substack{i,j,k,l=1\\i\neq j\neq k\neq l}}\sqrt{\lambda_i\lambda_j\lambda_k\lambda_l},
$
we have
\begin{align*}
&C^2_{l_1}(\psi)=\Br{(\sum^N_{i=1}{\sqrt{\lambda_i}})^2-1}^2\\
&=2\sum^N_{\substack{i,j=1\\i\neq j}}\lambda_i\lambda_j
+4\sum^N_{\substack{i,j,k=1\\i\neq j\neq k}}\lambda_i\sqrt{\lambda_j\lambda_k}+\sum^N_{\substack{i,j,k,l=1\\i\neq j\neq k\neq l}}\sqrt{\lambda_i\lambda_j\lambda_k\lambda_l}.\end{align*}
Thus,
\begin{align*}
&\mathbb{E}_\psi C^2_{l_1}(\psi)=\int d(\psi) C^2_{l_1}(\psi)\\
&=\Gamma(N)\int \Big[2\sum^N_{\substack{i,j=1\\i\neq j}}\lambda_i\lambda_j
+4\sum^N_{\substack{i,j,k=1\\i\neq j\neq k}}\lambda_i\sqrt{\lambda_j\lambda_k}\\
&~~~~~~~+\sum^N_{\substack{i,j,k,l=1\\i\neq j\neq k\neq l}}\sqrt{\lambda_i\lambda_j\lambda_k\lambda_l}\Big] \times\delta(1-\sum^N_{i=1}\lambda_i)\prod^N_{i=1}d\lambda_i\\
&=2\sum^N_{i\neq j} \frac{\Gamma(N)}{\Gamma(N+2)}
+4\sum^N_{i\neq j\neq k}\frac{\Gamma(N)}{\Gamma(N+2)}\Gamma(\frac{3}{2})^2\\
 &~~~~~~~+\sum^N_{i\neq j\neq k\neq l} \frac{\Gamma(N)}{\Gamma(N+2)}(\Gamma(\frac{3}{2}))^4\\
&=\frac{(N-1)(N-2)(N-3)}{16(N+1)}\pi^2+\frac{(N-1)(N-2)}{N+1}\pi\\
&~~~~~~~+\frac{2(N-1)}{N+1}.
\end{align*}
Also, we have
\begin{eqnarray*}
(\mathbb{E}_\psi C_{l_1}(\psi))^2=\frac{(N-1)^2}{16}\pi^2.
\end{eqnarray*}
Using above equations, we have
\begin{eqnarray*}
\mathbb{D}_\psi C_{l_1}(\psi)&= &\mathbb{E}_\psi C^2_{l_1}(\psi)
- (\mathbb{E}_\psi C_{l_1}(\psi))^2\\&=&
\frac{N-1}{16(N+1)}\Br{(16-5\pi)N +(7\pi^2-32\pi +32)}.
\end{eqnarray*}
The variance of the $l_1$ norm of coherence $\mathbb{D}_\psi C_{l_1}(\psi)\sim O(N)$ and therefore will tend to infinity as $N$ increases to infinity. The variance of $\frac{C_{l_1}(\psi)}{C^{max}_{l_1}}$ is given by
\begin{eqnarray}
\mathbb{D}_\psi \frac{C_{l_1}(\psi)}{C^{max}_{l_1}}
&=&\frac{1}{(N-1)^2}\mathbb{D}_\psi C_{l_1}(\psi)\nonumber\\
&=&\frac{1}{16(N^2-1)}\Br{(16-5\pi)N +(7\pi^2-32\pi +32)}.\nonumber
\end{eqnarray}
For $N\to\infty$, $\mathbb{D}_\psi \frac{C_{l_1}(\psi)}{C^{max}_{l_1}}$ goes to zero. As the variance  measures the spread of a set of data around the average (mean) value, the $l_1$ norm of coherence will not concentrate around its average value while the scaled $l_1$ norm of coherence $\frac{C_{l_1}(\psi)}{C^{max}_{l_1}}$ will concentrate around its average value as $N$ increases to infinity. In fact, let $Y$ be a random variable with expected value $\mathbb{E}Y$ and variance $\mathbb{D}Y$, then the Chebyshev's inequality \cite{Kallenberg2006},
\begin{eqnarray*}
\mathrm{Pr}(|Y-\mathbb{E}Y|>\varepsilon)\leq\frac{\mathbb{D}Y}{\varepsilon^2}
\end{eqnarray*}
implies
\begin{eqnarray}
\mathrm{Pr}\left(\abs{\frac{C_{l_1}(\psi)}{C^{max}_{l_1}}-\mathbb{E}_\psi\frac{C_{l_1}(\psi)}{C^{max}_{l_1}}}>\varepsilon\right)\leq\frac{\mathbb{D}_\psi \frac{C_{l_1}(\psi)}{C^{max}_{l_1}}}{\varepsilon^2}.
\end{eqnarray}
The right hand side will tend to zero as $N$ increases to infinity. This is in agreement with
Corollary \ref{cor:con_of_l1}.

\section{Distance of a random pure state from the set of maximally coherent states and its typicality}
\label{sec:aver_dis}
Given a state $\ket{\psi}=\sum^N_{i=1}\sqrt{\lambda_i}e^{i\theta_i}\ket{i}$, we ask how far this state is from the set $\cM$ of maximally coherent states. The maximally coherent states are given by $\ket{\phi}=\frac{1}{\sqrt{N}}\sum^N_{j=1}e^{i\theta_j}\ket{j}$ in the resource theory of coherence \cite{Baumgratz2014, Peng2015}. We first consider the trace distance and the Hilbert-Schmidt distance and then the Bures distance as the distance measures.
\subsection{The trace distance and the Hilbert-Schmidt distance}
From Eqs. (\ref{tr-dis}) and (\ref{hs-dis}), we have
\begin{align}
&\Delta_{\mathrm{Tr}}(\psi, \cM):=\inf_{\phi\in\cM}||\psi-\phi||_1\nonumber\\
&~~~~~~~~~~~~~~~~~~~~=2\sqrt{1-\frac{1}{N}\left(\sum^N_{i=1}\sqrt{\lambda_i}\right)^2};\nonumber\\
&\Delta_{\mathrm{HS}}(\psi, \cM):=\inf_{\phi\in\cM}||\psi-\phi||_2\nonumber\\
&~~~~~~~~~~~~~~~~~~~~~~=\sqrt{2}\sqrt{1-\frac{1}{N}\left(\sum^N_{i=1}\sqrt{\lambda_i}\right)^2}.\nonumber
\end{align}
Since $\Delta_{\mathrm{HS}}(\psi, \cM)$ is equivalent to $\Delta_{\mathrm{Tr}}(\psi, \cM)$ up to a factor $\frac{1}{\sqrt{2}}$, we only
need to  consider $\Delta_{\mathrm{Tr}}(\psi, \cM)$ here. From $C_{l_1}(\psi)=\left(\sum^N_{i=1}\sqrt{\lambda_i}\right)^2-1$, we have
\begin{eqnarray}
\Delta_{\mathrm{Tr}}(\psi, \cM)&=2\sqrt{1-\frac{1}{N}(C_{l_1}(\psi)+1)}\nonumber\\
&=\frac{2}{\sqrt{N}}\sqrt{N-1-C_{l_1}(\psi)}.
\end{eqnarray}
Thus, we have
\begin{eqnarray}\label{dual:1}
\frac{N\Delta^2_{\mathrm{Tr}}(\psi,\cM)}{4}+C_{l_1}(\psi)=C^{max}_{l_1}.
\end{eqnarray}
Similarly for Hilbert-Schmidt norm we have
\begin{eqnarray}\label{dual:2}
\frac{N\Delta^2_{\mathrm{HS}}(\psi,\cM)}{2}+C_{l_1}(\psi)=C^{max}_{l_1}.
\end{eqnarray}
Now, integrating both sides of Eq. \eqref{dual:1}, we get
\begin{eqnarray*}
\frac{N}{4}\int d(\psi)\Delta^2_{\mathrm{Tr}}(\psi,\cM)+\int d(\psi) C_{l_1}(\psi)=N-1.
\end{eqnarray*}
Defining $\Delta^{\mathrm{RMS}}_{\mathrm{Tr}}=\sqrt{\int d(\psi)\Delta^2_{\mathrm{Tr}}(\psi,\cM)}$, we have $\Delta^{\mathrm{RMS}}_{\mathrm{Tr}}=\sqrt{(4-\pi)\frac{N-1}{N}}$, which for $N\rightarrow\infty$ approaches to $\sqrt{4-\pi}\approx0.9625$. In the following, we show the concentration of the distance $\Delta^2_{\mathrm{Tr}}(\psi, \cM)$ around its average value.
\begin{prop}
\label{prop:tr-hs-con}
Let $\ket{\psi}$ be a random pure state in $\mathcal{H}$ with $dim\cH=N$. Then, the probability that $\Delta^2_{\mathrm{Tr}}(\psi, \cM)$ is not close to $\frac{N-1}{N}(4-\pi)$ is bounded from above by an exponentially small number when $N$ is large, i.e., for any $\varepsilon>0$
\begin{eqnarray}
\mathbb{P}_{\mathrm{Tr}} \leq 2\exp\Pa{-\frac{N\varepsilon^2}{36\pi^3\ln2}},
\end{eqnarray}
where $\mathbb{P}_{\mathrm{Tr}}=\mathrm{Pr} \Set{\left|\Delta^2_{\mathrm{Tr}}(\psi, \cM)- \frac{N-1}{N}(4-\pi) \right|>\varepsilon}$.
\end{prop}

\begin{proof}
Consider the functional $\cG:\ket{\psi}\to \cG(\psi)=\Delta^2_{\mathrm{Tr}}(\psi, \cM)$. Then,
\begin{eqnarray*}
|\Delta^2_{\mathrm{Tr}}(\psi, \cM)-\Delta^2_{\mathrm{Tr}}(\phi, \cM)|&=&\frac{4}{N}|C_{l_1}(\psi)-C_{l_1}(\phi)|\\
&\leq& 2\sqrt{2}||\psi-\phi||_2,
\end{eqnarray*}
where in the last line, we have used Eq. (\ref{eq:lip-l1}). Then, Lipschitz constant for $\Delta^2_{\mathrm{Tr}}(\psi, \cM)$ is $\eta\leq 2\sqrt{2}$. Now, from L\'evy's lemma, the proof of the proposition follows.
\end{proof}

\subsection{Bures distance}
The Bures distance between a random pure state $\ket{\psi}=\sum^N_{i=1}\sqrt{\lambda_i}e^{i\theta_i}\ket{i}$ and the set $\cM$ of maximally coherent states is given by
\begin{eqnarray}
\Delta_\mathrm{B}(\psi, \cM):&=&\inf_{\phi\in\cM}D_\mathrm{B}(\psi,\phi)\nonumber\\
&=&\sqrt{2}\sqrt{1-\frac{1}{\sqrt{N}}\left(\sum^N_{i=1}\sqrt{\lambda_i}\right)}.
\end{eqnarray}
Let us calculate $\mathbb{E}_\psi \Delta^2_\mathrm{B}(\psi, \cM)=\int d(\psi)\Delta^2_\mathrm{B}(\psi, \cM)$.
\begin{align*}
&\int d(\psi)\Delta^2_\mathrm{B}(\psi, \cM)\\
&=2-\frac{2}{\sqrt{N}}\int \Gamma(N) (\sum^N_{i=1}\sqrt{\lambda_i})\delta(1-\sum^N_{i=1}\lambda_i)\prod^N_{i=1}d\lambda_i\\
&=2-\frac{2}{\sqrt{N}}\Gamma(N)\sum^N_{i=1}\int \sqrt{\lambda_i}\delta(1-\sum^N_{i=1}\lambda_i)\prod^N_{i=1}d\lambda_i\\
&=2-\frac{2}{\sqrt{N}}\Gamma(N)\times N\times \frac{\Gamma(\frac{3}{2})}{\Gamma(N+\frac{1}{2})}\\
&=2-\sqrt{N}\cdot B\left(\frac{1}{2}, N\right),
\end{align*}
where the beta function $B(x,y)$ is defined as $B(x,y):=\int^1_0 t^{x-1}(1-t)^{y-1}dt$ for $\mathrm{Re}~ x>0$ and $\mathrm{Re}~ y>0$ \cite{Andrews1999}. Moreover, $B(x,y)$ can be written as $B(x,y)=\frac{\Gamma(x)\Gamma(y)}{\Gamma(x+y)}$ \cite{Andrews1999}.
In the limit $N\to \infty$, $\Gamma(N)\sim \sqrt{2\pi}N^{N-1/2}e^{-N}$ \cite{Andrews1999}. Then
$\sqrt{N}\cdot B(\frac{1}{2}, N)\to \sqrt{\pi}$, which implies $\mathbb{E}_\psi \Delta^2_\mathrm{B}(\psi, \cM)\to 2-\sqrt{\pi}\approx  0.2275$ as $N\to\infty$. Moreover, similar to Proposition \ref{prop:tr-hs-con}, we have the following proposition.
\begin{prop}
\label{prop:bures-con}
Let $\ket{\psi}$ be random pure state in $\mathcal{H}$ with $dim\cH=N$. Then, the probability that $\Delta^2_\mathrm{B}(\psi, \cM)$ is not close to $2-\sqrt{N} B(\frac{1}{2}, N)$ is bounded from above by an exponentially small number when $N$ is large, i.e., for any $\varepsilon>0$
\begin{eqnarray}
\mathbb{P}_\mathrm{B }\leq 2\exp\Pa{-\frac{N\varepsilon^2}{9\pi^3\ln2}},
\end{eqnarray}
where $\mathbb{P}_\mathrm{B}=\mathrm{Pr} \Set{\left|\Delta^2_\mathrm{B}(\psi, \cM)- \left(2-\sqrt{N} B(\frac{1}{2}, N) \right)\right|>\varepsilon}$.
\end{prop}

\begin{proof}
Consider the map $\cG:\ket{\psi}\to \Delta^2_\mathrm{B}(\psi, \cM)$. Then, similar to the
$\sD^2_\mathrm{B}(\psi_{AB},\cM_E)$  case,
\begin{align*}
\left|\Delta^2_\mathrm{B}(\psi, \cM)-\Delta^2_\mathrm{B}(\phi, \cM)\right|\leq\sqrt{2}||\psi-\phi||_2.
\end{align*}
Thus, the Lipschitz constant for the function $\cG:\ket{\psi}\to \Delta^2_\mathrm{B}(\psi, \cM)$ can be taken to be $\sqrt{2}$.
Now, applying L\'evy's Lemma to the function $\cG:\ket{\psi}\to \Delta^2_\mathrm{B}(\psi, \cM)$, the proof of the proposition follows.
\end{proof}
While it is natural to consider the distance between a random pure state $\ket{\psi}$ and the set $\cM$ of maximally coherent states, the distance between the diagonal part of a random pure state and the maximally mixed state can also be viewed as a quantifier to measure how far $\ket{\psi}$  is away from   the maximally coherent state. Here, the diagonal part of a random pure state $\ket{\psi}=\sum^N_{i=1}\sqrt{\lambda_i}e^{i\theta_i}\ket{i}$ is denoted by $\psi^\md=\mathrm{diag}\set{\lambda_1, \lambda_2,\ldots, \lambda_N}$. In Ref. \cite{UttamS2015}, an explicit formula for the average trace distance between the diagonal part $\psi^\md$ of a random pure state $\ket{\psi}$ and the maximally mixed state $\frac{\mathbb{I}}{N}$ is obtained, and is given by
\begin{eqnarray}
\mathbb{E}_{\psi} \norm{\psi^\md-\frac{\mathbb{I}}{N}}_1=2\left(1-\frac{1}{N}\right)^N.
\end{eqnarray}
In limit $N\to\infty$, $\mathbb{E}_{\psi} \norm{\psi^\md-\frac{\mathbb{I}}{N}}_1 \to \frac{2}{e}$. We wonder if there is any connection between the quantity $\norm{\psi^\md-\frac{\mathbb{I}}{N}}_1$ and the three distances that we have defined in the preceding paragraphs. Similar to the entanglement case, we find that   $\Delta^2_\mathrm{B}(\psi, \cM)$ is a lower bound of $\norm{\psi^\md-\frac{\mathbb{I}}{N}}_1$.
Using Eq. (\ref{lem:1v2}) with $\rho=\psi^\md$ and $\sigma=\frac{\mathbb{I}}{N}$, we obtain
\begin{eqnarray*}
\norm{\psi^\md-\frac{\mathbb{I}}{N}}_1 \geq \Delta^2_\mathrm{B}(\psi, \cM).
\end{eqnarray*}
Thus,
\begin{eqnarray*}
\mathbb{E}_{\psi}\norm{\psi^\md-\frac{\mathbb{I}}{N}}_1\geq \mathbb{E}_{\psi}\Delta^2_\mathrm{B}(\psi, \cM).
\end{eqnarray*}
As we have elaborated, the average distances between a random pure state and the set of the maximally coherent states, namely, $\Delta^{\mathrm{RMS}}_{\mathrm{Tr}}$ (similarly $\Delta^{\mathrm{RMS}}_{\mathrm{HS}}$ and $\Delta^{\mathrm{RMS}}_{\mathrm{B}}$) and $\mathbb{E}_{\psi} \norm{\psi^\md-\frac{\mathbb{I}}{N}}_1$ do not tend to zero in the limit $N\to\infty$. This is consistent with the fact that $\frac{\mathbb{E}_\psi C_{l_1}(\psi)}{C^{max}_{l_1}}=\frac{\pi}{4}$ for any dimension $N$.

\section{Average $\alpha$-classical purity and its typicality }
\label{aver_puri}
The classical purity $P_{\mathrm{cl}}(\rho)$ with respect to a fixed measurement basis $\{\ket{i}\}_i$ of a state $\rho$ is defined as $P_{\mathrm{cl}}(\rho):= \mathrm{Tr}[(\rho^\md)^2]$, where $\rho^\md=\sum_i \bra{i}\rho\ket{i}\ket{i}\bra{i}$ is the diagonal part of $\rho$ in the basis $\{\ket{i}\}_i$. The classical purity has been shown to be related to the $l_1$ norm of coherence \cite{Cheng2015}. Moreover, in catalytic coherence transformations as in Ref. \cite{Bu2015b} the functions of the form $P_{\mathrm{cl}}^\alpha(\rho):= \mathrm{Tr}[(\rho^\md)^\alpha]$ appear naturally. We call $P_{\mathrm{cl}}^\alpha(\rho)$ $\alpha$-classical purity of $\rho$ and it can be seen as a generalization of the classical purity.
Here we calculate the expected value of the $\alpha$-classical purity of pure states for  $\alpha\in(0,1)\cup(1,+\infty)$. Let $\psi^\md =\{\lambda_1,\ldots\lambda_N\}$ denotes the diagonal part of the pure state $\ket{\psi}$ in a fixed basis, then the expected $\alpha$-classical purity is given by
\begin{eqnarray}
\mathbb{E}_{\psi}\mathrm{Tr}[(\psi^\md)^{\alpha}]&=&\sum^N_{i=1} \int p_c(\mathbf{\Lambda}) \lambda^{\alpha}_i
d\mathbf{\Lambda}\nonumber\\
&=&\frac{\Gamma(\alpha+1)\Gamma(N+1)}{\Gamma(\alpha+N)}.
\end{eqnarray}
Here $p_c(\mathbf{\Lambda})d\mathbf{\Lambda}$ is given by Eq. (\ref{eq:prob-dis}). For $\alpha>1$, $\mathbb{E}_{\psi}\mathrm{Tr}[(\psi^\md)^{\alpha}]$ tends
to zero as $N$ increases to infinity. Moreover, amazingly, we will show that $\mathrm{Tr}[(\psi^\md)^\alpha]$ will concentrate to $0$ for $\alpha>1$ as $N\rightarrow\infty$. First, the variance of $\mathrm{Tr}[(\psi^\md)^\alpha]$
can be expressed as
\begin{eqnarray*}
\mathbb{D}_{\psi}\mathrm{Tr}[(\psi^\md)^{\alpha}]
=\mathbb{E}_{\psi}\Pa{\mathrm{Tr}[(\psi^\md)^{\alpha}]}^2-\Pa{\mathbb{E}_{\psi}\mathrm{Tr}[(\psi^\md)^{\alpha}]}^2,
\end{eqnarray*}
and
\begin{eqnarray*}
&&\mathbb{E}_{\psi}\Pa{\mathrm{Tr}[(\psi^\md)^{\alpha}]}^2\\
&&=\sum_i \int p_c(\mathbf{\Lambda}) \lambda^{2\alpha}_i
d\mathbf{\Lambda}+
\sum_{i\neq j}\int p_c(\mathbf{\Lambda}) \lambda^\alpha_i\lambda^\alpha_j
d\mathbf{\Lambda}\\
&&=\frac{\Gamma(N+1)}{\Gamma(N+2\alpha)}[\Gamma(2\alpha+1)+(N-1)\Gamma(\alpha+1)^2].
\end{eqnarray*}
Thus
\begin{eqnarray*}
\mathbb{D}_{\psi}\mathrm{Tr}[(\psi^\md)^{\alpha}]
&=&\frac{\Gamma(N+1)}{\Gamma(N+2\alpha)}[\Gamma(2\alpha+1)+(N-1)\Gamma(\alpha+1)^2]\\
&&-
\Gamma(N+1)^2\frac{\Gamma(\alpha+1)^2}{\Gamma(\alpha+N)^2}
\end{eqnarray*}
It is easy to see that for $\alpha>1$, $\mathbb{D}_{\psi}\mathrm{Tr}[(\psi^\md)^{\alpha}]\to 0,~\mathrm{as}~N~\to \infty$. Thus, by the Chebyshev inequality \cite{Kallenberg2006}, $\mathrm{Tr}[(\psi^\md)^{\alpha}]$ concentrates to $0$ as $N\to\infty$. However, for $0<\alpha<1$
$\mathbb{D}_{\psi}\mathrm{Tr}[(\psi^\md)^{\alpha}]\to\infty$ as $N\to \infty$.
Thus, as the variance  measures the spread of a set of data around the average (mean) value,  it seems that there is no concentration phenomenon  for $\alpha$-classical purity in $\alpha\in(0,1)$.

\section{Conclusion}\label{sec:con}
In this work, we have investigated the average distance of Haar distributed pure quantum states from the set of maximally entangled and coherent states in order to delineate the typical properties of the average entanglement and average $l_1$ norm of coherence. The average distance of Haar distributed bipartite pure quantum states from the set of maximally entangled states is typically nonzero in the limit of larger Hilbert space dimensions which is unexpected. We further demonstrate the typicality of $l_1$ norm of  coherence scaled by the maximum value of the $l_1$ norm of coherence by finding $\mathbb{E}_\psi C_{l_1}(\psi)/C^{max}_{l_1}=\frac{\pi}{4}$ and utilizing L\'evy's lemma. This is in contrast to the case of relative entropy of coherence for which $\mathbb{E}_\psi C_r(\psi)/C^{max}_r\sim 1 ~\mathrm{as}~N\to\infty$, where $C_r(\psi)=-\mathrm{Tr}\left(\psi^{(\mathrm{d})}\ln \psi^{(\mathrm{d})}\right)$ with $\psi^{(\mathrm{d})}$ being the diagonal part of $\ket{\psi}$ i
 n the gi
 ven reference basis.

Here we have obtained the average $l_1$ norm of coherence of random pure states, however, the average $l_1$ norm of coherence for random mixed states is still unknown and it will be interesting to obtain it.  Also, for pure states, we have following relationship $C_{l_1}(\psi)\geq \max\set{C_r(\psi), 2^{C_r(\psi)}-1}$ \cite{Swapan2016} between the $l_1$ norm of coherence and the relative entropy of coherence, but we don't know whether this relationship holds in the case of mixed states too. It will be insightful to pursue this problem in future.

\smallskip
\noindent
\begin{acknowledgments}
U.S. acknowledges a research fellowship of Department of Atomic Energy, Government of India. L.Z. acknowledges the National Natural Science
Foundation of China (No.11301124) for support. J.W. is supported by the Natural Science Foundations of China (Grants No.11171301 and No. 10771191) and the Doctoral Programs Foundation of the Ministry of Education of China (Grant No. J20130061).
\end{acknowledgments}

\bibliographystyle{apsrev4-1}
 \bibliography{aver-coh-lit}

\end{document}